\documentclass[journal,10pt]{IEEEtran}
\usepackage{graphicx,psfrag,subfigure,url,amsmath,amsfonts,
epsfig,amssymb}

\newtheorem{claim}{Claim}
\newtheorem{lemma}{Lemma}
\newtheorem{remark}{Remark}
\newtheorem{bound}{Bound}

\long\def\symbolfootnote[#1]#2{\begingroup
\def\thefootnote{\fnsymbol{footnote}}\footnote[#1]{#2}\endgroup}

\begin{document}

%opening
\title{An outer bound for 2-receiver discrete memoryless broadcast channels}
\author{Chandra Nair \\ ~ \\ Department of Information Engineering \\  The Chinese University of Hong Kong}

\maketitle

\begin{abstract}
An outer bound to the discrete memoryless broadcast channel is presented. We compare it to the known outer bounds and show that the outer bound presented is at least as tight as the existing bounds.
\end{abstract}

\section{Introduction}
\label{se:intro}

There has been a series of outer bounds presented to the capacity region of the broadcast channel \cite{nae07,lik07,lks08}. All the bounds follow from the use of Fano's inequality and the Csiszar sum lemma\cite{czk78}. In this note, we present another outer bound along these lines that is at least as tight as the known bounds.

\section{Two receiver broadcast channel with private messages only}

The following lemma presents an outer bound for the capacity region of the two receiver discrete memoryless broadcast channels.  

\begin{lemma}
\label{le:bound}
Consider the set of all random variables $U,V,W_1,W_2$ such that $(U,V,W_1,W_2) \to X \to (Y_1,Y_2)$ for a Markov chain. Further assume that $U$ and $V$ are independent; and the distribution $(U,V,W_1,W_2,X,Y_1,Y_2)$ satisfies the following equalities:
\begin{align}
I(U;Y_1|W_1) &= I(U;Y_1|V,W_1) \nonumber \\
I(V;Y_2|W_2) &= I(V;Y_2|U,W_2) \nonumber\\
I(U;V|W_1,W_2,Y_1) &= I(U;V|W_1,W_2,Y_2) \nonumber \\
I(W_2;Y_1|W_1) &= I(W_1;Y_2|W_2) \nonumber\\
I(W_2;Y_1|U,W_1) &= I(W_1;Y_2|U,W_2) \label{eq:auxeq}\\
I(W_2;Y_1|V,W_1) &= I(W_1;Y_2|V,W_2) \nonumber\\
I(W_2;Y_1|U,V,W_1) &= I(W_1;Y_2|U,V,W_2) \nonumber.
\end{align}
Then the set of rate pairs $R_1,R_2$ satisfying
\begin{align*}
R_1 &\leq I(U;Y_1|W_1) \\
R_2 &\leq I(V;Y_2|W_2)
\end{align*}
constitutes an outer bound to the capacity region of the discrete memoryless broadcast channel.
\end{lemma}

\begin{proof}
The inequalities follows immediately from  Fano's inequality and the following identifications:
\begin{align*}
\hat{W}_{1i} &= Y_1^{i-1} \\
\hat{W}_{2i} &= Y_{2~i+1}^n \\
U & = M_1 \\
V &= M_2.
\end{align*}
We then set $W_{1} = (\hat{W}_1,Q), W_{2} = (\hat{W}_2,Q)$, where $Q$ is an independent random variable chosen uniformly at random from the interval $\{1,...,n\}$.

The last four equalities are a direct application of the Csiszar sum lemma \cite{czk78} and the proof is omitted. The first two equalities follow from Fano's inequality and the independence of $M_1$ and $M_2$; and the proof is again omitted. The third equality follows as follows:
\begin{align*}
& I(U;V|W_1, W_2, Y_1) - I(U;V|W_1,W_2,Y_2) \\
&\quad = \lim_{n \to \infty} \frac 1n \sum_{i=1}^n I(M_1;M_2|Y_1^{i},Y_{2~i+1}^n) - I(M_1;M_2|Y_1^{i-1}, Y_{2~i}^n) \\
&\quad = \lim_{n \to \infty} \frac 1n \left( I(M_1;M_2|Y_1^n) - I(M_1;M_2|Y_2^n) \right) \\
&\quad = 0
\end{align*}
The last step follows from  Fano's inequality. 
\end{proof}

\begin{remark}
We note the following divergence from the normal presentation of the outer bounds: the absence of a sum rate constraint, as well as the presence of a number of equalities.
\end{remark}

We will compare this bound to the following existing bound\footnote{The equivalence of the bounds can be observed from the fact that for the identifications in \cite{nae07} $I(U;V|W,Y_1) = I(U;V|W,Y_2)$, and this implies the bound presented in \cite{lik07}.} for the same setting.
\begin{bound}
\label{bd:outer}
The union of rate pairs $(R_1,R_2)$  that satisfy the following inequalities
\begin{align*}
R_1 &\leq I(U,W;Y_1) \\
R_2 &\leq I(V,W;Y_2) \\
R_1 + R_2 &\leq \min\{I(W;Y_1), I(W;Y_2)\} + I(U;Y_1|W) \\
&\qquad + I(V;Y_2|U,W) \\
R_1 + R_2 &\leq \min\{I(W;Y_1), I(W;Y_2)\} + I(U;Y_1|V,W) \\
&\qquad + I(V;Y_2|W).
\end{align*}
over all $p(u)p(v)p(w|u,v)p(x|u,v,w)p(y_1,y_2|x)$ forms an outer bound to the capacity region.
\end{bound}

\begin{claim}
\label{cl:wk1}
The region specified by the lemma \ref{le:bound} is at least as tight as the region specified by Bound \ref{bd:outer}.
\end{claim}
\begin{proof}
We need to show that any $(R_1,R_2)$ satisfying the constraints of Lemma  \ref{le:bound} is contained in the region described by Bound \ref{bd:outer}.
To show the inclusion, we set $W=(W_1,W_2)$. 

Observe that
\begin{align*}
&I(V;Y_2|U,W_1,W_2) = I(V;Y_2|W_1,W_2) \\
&\qquad \qquad  - I(U;V|W_1,W_2) + I(V;U|W_1,W_2,Y_2) 
\end{align*}
Using the equality $$I(V;U|W_1,W_2,Y_2) = I(V;U|W_1,W_2,Y_1)$$ it is easy to see that
\begin{align}
& I(U;Y_1|W_1,W_2) + I(V;Y_2|U,W_1,W_2) \nonumber \\
& = I(U;Y_1|V,W_1,W_2) + I(V;Y_2|W_1,W_2). \label{eq:obs1}
\end{align}
Therefore the two sum rate constraints in Lemma \ref{bd:outer} are identical. 

Hence it  suffices to prove that
\begin{align*}
& I(U;Y_1|W_1) + I(V;Y_2|W_2) \\
&\quad \leq I(W_1,W_2;Y_1) + I(U;Y_1|W_1,W_2) \\
&\qquad + I(V;Y_2|U, W_1,W_2).
\end{align*}
(The other one obtained by replacing $I(W_1,W_2;Y_1)$ with $I(W_1,W_2;Y_2)$ follows similarly. To get the symmetric expression, just use \eqref{eq:obs1}.)

Observe that
\begin{align*}
& I(U,W_1,W_2;Y_1) +  I(V;Y_2|U,W_1,W_2) \\
&= I(W_1;Y_1) + I(U;Y_1|W_1) + I(W_2;Y_1|U,W_1) \\
& \quad + I(V;Y_2|U, W_1,W_2) \\ 
& = I(W_1;Y_1) + I(U;Y_1|W_1) + I(W_1;Y_2|U,W_2) \\
& \quad + I(V;Y_2|U, W_1,W_2)  \\
& = I(W_1;Y_1) + I(U;Y_1|W_1) + I(V, W_1;Y_2|U,W_2) \\
& = I(W_1;Y_1) + I(U;Y_1|W_1) + I( W_1;Y_2|U,V,W_2) \\
& \quad + I(V;Y_2|U,W_2) \\
& \stackrel{(a)}{=} I(W_1;Y_1) + I(U;Y_1|W_1) + I( W_1;Y_2|U,V,W_2) \\
& \quad + I(V;Y_2|W_2) \\
& \geq I(U;Y_1|W_1)  + I(V;Y_2|W_2),
\end{align*}
where $(a)$ follows from the following:
\begin{equation*}
 I(V;Y_2|U,W_2)  = I(V;Y_2|W_2).
\end{equation*}
\end{proof}

\section{Two receiver broadcast channel with common message as well as private messages}
\label{se:com}

The following outer bound was presented in \cite{lks08} for the capacity region of the broadcast channel for two receivers with a common message as well as private messages.

\begin{bound} \cite{lks08} The capacity region is a subset of the {\em New-Jersey} region, which can be obtained by taking the union of rate triples $(R_0,R_1,R_2)$ satisfying
\begin{align*}
R_0 &\leq \min{I(T;Y_1|W_1), I(T;Y_2|W_2)} \\
R_1 & \leq I(U;Y_1|W_1) \\
R_2 &\leq I(V;Y_2|W) \\
R_0 + R_1 &\leq I(T,U;Y_1|W_1) \\
R_0 + R_1 & \leq I(U;Y_1|T,W_1,W_2) + I(T,W_1;Y_2|W_2) \\
R_0 + R_2 &\leq I(T,U;Y_2|W_2) \\
R_0 + R_2 & \leq I(V;Y_2|T,W_1,W_2) + I(T,W_2;Y_1|W_1) \\
R_0 + R_1 + R_2 &\leq I(U;Y_1|T,V,W_1,W_2) + I(T, V, W_1;Y_2|W_2) \\
R_0 + R_1 + R_2 &\leq I(V;Y_2|T,U,W_1,W_2) + I(T, U, W_2;Y_1|W_1) \\
R_0 + R_1 + R_2 &\leq I(U;Y_1|T,V,W_1,W_2) + I(T,W_1,W_2;Y_1) \\
&\quad + I(V;Y_2|T,W_1,W_2) \\
R_0 + R_1 + R_2 &\leq I(V;Y_2|T,U,W_1,W_2) + I(T,W_1,W_2;Y_2) \\
&\quad + I(U;Y_1|T,W_1,W_2) 
\end{align*}
for some $p(u)p(v)p(t)p(w_1,w_2|u,v,t)p(x|u,v,t,w_1,w_2)p(y_1,y_2|x)$. Further one can restrict $X$ to be a deterministic function of $(u,v,t,w_1,w_2)$ and also assume that the marginals of $U,V,T$ are uniform.
\end{bound}

Similar to lemma \ref{le:bound} we can write an outer bound for this case as well, and this region is at least as tight as the {\em New Jersey} outer bound.

\begin{lemma}
\label{le:bound1}
Consider the set of all random variables $T,U,V,W_1,W_2$ such that $(T,U,V,W_1,W_2) \to X \to (Y_1,Y_2)$ for a Markov chain. Further assume that $T$,$U$, and $V$ are independent; and the distribution $(U,V,W_1,W_2,X,Y_1,Y_2)$ satisfies the following equalities:
\begin{align}
I(T;Y_1|W_1) &= I(T;Y_2|W_2) \nonumber \\
I(T;Y_1|W_1) &= I(T;Y_1|V,W_1) = I(T;Y_1|U, W_1) \nonumber \\
&= I(T;Y_1|U,V,W_1) \nonumber \\
I(T;Y_2|W_2) &= I(T;Y_2|V,W_2) = I(T;Y_2|U, W_2) \label{eq:auxeq1} \\
& = I(T;Y_2|U,V,W_2) \nonumber \\
I(U;Y_1|W_1) &= I(U;Y_1|V,W_1) = I(U;Y_1|T, W_1) \nonumber\\
& = I(U;Y_1|T,V,W_1) \nonumber \\
I(V;Y_2|W_2) &= I(V;Y_2|U,W_2) =  I(V;Y_2|T, W_2) \nonumber \\
& = I(V;Y_2|T,U,W_2), \nonumber 
\end{align}
\begin{equation}
\label{eq:auxeq2}
I(B_1;B_2|A,W_1,W_2,Y_1) = I(B_1;B_2|A,W_1,W_2,Y_2) 
\end{equation}
holds for all $A \subseteq \{T,U,V\}, B_1 \subseteq\{T,U\}, B_2\subseteq\{T,V\}$, and
\begin{equation}
\label{eq:auxeq3}
I(W_2;Y_1|A,W_1) = I(W_1;Y_2|A,W_2)
\end{equation}
holds for all $A \subseteq \{T,U,V\}$.

Then the set of rate tuples $(R_0,R_1,R_2)$ satisfying
\begin{align*}
R_0 &\leq \min\{I(T;Y_1|W_1), I(T;Y_2|W_2)\}\\
R_1 &\leq I(U;Y_1|W_1) \\
R_2 &\leq I(V;Y_2|W_2)
\end{align*}
constitutes an outer bound to the capacity region of the discrete memoryless broadcast channel.

Further, just as in Lemma \ref{le:bound1} one can restrict $X$ to be a deterministic function of $(u,v,t,w_1,w_2)$ and also assume that the marginals of $U,V,T$ are uniform.
\end{lemma}

\begin{proof}
$T=M_0$ is the only new identification as compared to Lemma \ref{le:bound}. The arguments for this lemma are similar to those of Lemma \ref{le:bound} and are omitted.
\end{proof}

\begin{claim}
\label{cl:wk2}
The region presented by Lemma \ref{le:bound1} is at least as tight as the {\em New-Jersey} outer bound.
\end{claim}
\begin{proof}
Again the arguments are similar to those of Claim \ref{cl:wk1} and are omitted.
\end{proof}

\section{Conclusion}
\label{se:conc}
An outer bound to the capacity region to the two receiver broadcast channel (with and without common information) is determined. In both cases, this is at least as tight as the currently best known bounds. 

\section*{Acknowledgement}
The author wishes to thank Prof. Bruce Hajek as this work benefitted tremendously from the discussions between the author and Prof. Hajek, during Prof. Hajek's visit to the Chinese University of Hong Kong in January 2008.

\bibliographystyle{IEEEtran}
\bibliography{mybiblio}

\begin{thebibliography}{1}
\providecommand{\url}[1]{#1}
\csname url@samestyle\endcsname
\providecommand{\newblock}{\relax}
\providecommand{\bibinfo}[2]{#2}
\providecommand{\BIBentrySTDinterwordspacing}{\spaceskip=0pt\relax}
\providecommand{\BIBentryALTinterwordstretchfactor}{4}
\providecommand{\BIBentryALTinterwordspacing}{\spaceskip=\fontdimen2\font plus
\BIBentryALTinterwordstretchfactor\fontdimen3\font minus
  \fontdimen4\font\relax}
\providecommand{\BIBforeignlanguage}[2]{{%
\expandafter\ifx\csname l@#1\endcsname\relax
\typeout{** WARNING: IEEEtran.bst: No hyphenation pattern has been}%
\typeout{** loaded for the language `#1'. Using the pattern for}%
\typeout{** the default language instead.}%
\else
\language=\csname l@#1\endcsname
\fi
#2}}
\providecommand{\BIBdecl}{\relax}
\BIBdecl

\bibitem{czk78}
I.~Csiz\'{a}r and J.~K{\"{o}}rner, ``Broadcast channels with confidential
  messages,'' \emph{IEEE Trans. Info. Theory}, vol. IT-24, pp. 339--348, May,
  1978.

\bibitem{nae07}
C.~Nair and A.~El~Gamal, ``An outer bound to the capacity region of the
  broadcast channel,'' \emph{IEEE Trans. Info. Theory}, vol. IT-53, pp.
  350--355, January, 2007.

\bibitem{lik07}
Y.~Liang and G.~Kramer, ``Rate regions for relay broadcast channels,''
  \emph{IEEE Transactions on Information Theory}, vol.~53, no.~10, pp.
  3517--3535, 2007.

\bibitem{lks08}
Y.~Liang, G.~Kramer, and S.~Shamai, ``Capacity outer bounds for relay broadcast
  channels,'' \emph{Proceedings of IEEE Inf. Theory Workshop, Porto, Portugal},
  pp. 2--4, 2008.

\end{thebibliography}
%\bibliography{C:/Users/cnair/Desktop/documents/tex/mybiblio}

\end{document}